\documentclass[11pt]{article}

\usepackage[a4paper, left=1in, right=1in, top=1in, bottom=1in]{geometry}
\usepackage{cmap} % Load before fontenc 
\usepackage[utf8]{inputenc}
\usepackage[english]{babel}
\usepackage[T1]{fontenc}

\usepackage{amsfonts}
\usepackage{amsmath}
\usepackage{amssymb}
\usepackage{amstext}
\usepackage{amscd}
\usepackage{amsthm}
\usepackage{mathrsfs}
\usepackage{mathtools}
\usepackage{longdivision}
\usepackage{nicefrac}
\usepackage{polynom}
\usepackage{bbm}
\usepackage{complexity}
\usepackage{dsfont}
\usepackage{physics}
\usepackage{braket}

\usepackage{xcolor}
\definecolor{blueviolet}{rgb}{0.2, 0.2, 0.6}
\definecolor{webgreen}{rgb}{0,.5,0}
\definecolor{webbrown}{rgb}{.6,0,0}
\usepackage[
  bookmarks=false,
  colorlinks=true, %allcolors=blueviolet,
  urlcolor=webbrown,
  linkcolor=blueviolet, 
  citecolor=webgreen,
  pdfstartpage=1,
  pdfstartview={FitH},  % FitBH
  bookmarksopen=false
  ]{hyperref}
\usepackage{zref-clever} % automatically resolve references to include environment name
\zcsetup{cap, nameinlink}
\newcommand{\cref}[1]{\zcref{#1}} % aliases like cleveref
\newcommand{\Cref}[1]{\zcref[S]{#1}}

\usepackage{appendix}
\usepackage{url}
\usepackage{listings}
\usepackage{tikz}
\usepackage{float}
\usepackage{makeidx}
\usepackage[shortlabels]{enumitem}
\usepackage[linesnumbered,ruled,vlined]{algorithm2e}
\zcRefTypeSetup{algocf}{
Name-sg = Algorithm ,
name-sg = algorithm ,
Name-pl = Algorithms ,
name-pl = algorithms ,
}
\usepackage[normalem]{ulem}
\usepackage{booktabs}

\allowdisplaybreaks

\DeclareMathOperator*{\dtr}{d_{\mathrm{tr}}}

\newcommand{\KeyGen}{\mathsf{KeyGen}}
\newcommand{\StateGen}{\mathsf{StateGen}}

\newcommand{\Ver}{\mathsf{Ver}}

\newcommand{\negl}{\mathsf{negl}}

\newcommand{\EFI}{\mathsf{EFI}}

\numberwithin{equation}{section}

\newcommand{\nocontentsline}[3]{}
\let\origcontentsline\addcontentsline
\newcommand\stoptoc{\let\addcontentsline\nocontentsline}
\newcommand\resumetoc{\let\addcontentsline\origcontentsline}

\newtheorem{theorem}{Theorem}
\newtheorem{prop}{Proposition}

\newtheorem{corollary}{Corollary}
\newtheorem{definition}{Definition}

\zcRefTypeSetup{prop}{
Name-sg = Proposition ,
name-sg = proposition ,
Name-pl = Propositions ,
name-pl = propositions ,
}
\zcRefTypeSetup{claim}{
Name-sg = Claim ,
name-sg = claim ,
Name-pl = Claims ,
name-pl = claims ,
}
\zcRefTypeSetup{remark}{
Name-sg = Remark ,
name-sg = remark ,
Name-pl = Remarks ,
name-pl = remarks ,
}

\newcommand{\stkout}[1]{\ifmmode\text{\sout{\ensuremath{#1}}}\else\sout{#1}\fi}
\newif\ifverbose
\verbosetrue
\begin{document}

\title{Equivalence Between Average-Case Hardness of Learning and Cryptography for Mixed Quantum States}

\author{Alexandru Cojocaru\thanks{Authors are listed alphabetically.}\;\thanks{University of Edinburgh. Email: \href{mailto:a.cojocaru@ed.ac.uk}{a.cojocaru@ed.ac.uk}} \and Laura Lewis\footnotemark[1]\;\thanks{University of Edinburgh. UC Berkeley. Email: \href{mailto:lllewis@berkeley.edu}{lllewis@berkeley.edu}}}

\date{}

\maketitle

\begin{abstract}
The relationship between cryptography and learning theory has long been a central theme in the foundations of theoretical computer science: cryptographic primitives can imply hardness of learning, while hardness of learning can in turn be used to construct cryptographic schemes. Recent works have begun exploring analogous connections in the quantum setting, relating the average-case hardness of learning quantum states (AHL) to cryptographic primitives such as one-way state generators (OWSG).
Despite recent progress exploring this for pure states, the relationship for mixed states has remained an open question.

In this work, we prove that the existence of AHL for mixed quantum states is equivalent to the existence of inefficiently verifiable one-way state generators (IV-OWSGs).
As a consequence, this relates mixed-state AHL to EFI pairs.
Moreover, as a corollary of existing results, we obtain a separation between IV-OWSGs and OWSGs relative to the SWAP oracle.
\end{abstract}

\section{Introduction}

Classically, there is a fundamental connection between cryptography and learning theory, where they are often viewed as each other's antithesis~\cite{rivest1991cryptography}.
Namely, cryptographic constructions can imply the impossibility of efficient learning~\cite{kearns1994cryptographic,klivans2009cryptographic,song2021cryptographic,daniely2014average,valiant1984theory,hirahara2023learning,naor2006learning,angluin1988queries,naor2015bloom}, and the hardness of learning can be used for novel constructions of cryptographic primitives~\cite{impagliazzo1990no,oliveira2016conspiracies,blum1993cryptographic,regev2009lattices,gentry2008trapdoors,lyubashevsky2010ideal,micciancio2009lattice}.
For instance, early results show the intractability of PAC learning algorithms for classes of Boolean functions based on constructions of public-key cryptosystems~\cite{kearns1994cryptographic,klivans2009cryptographic}.
Meanwhile, a prominent example in the opposite direction is the Learning with Errors problem, which is widely believed to be computationally intractable to solve, even for quantum computers, leading to many (post-quantum) cryptographic schemes based on it~\cite{regev2009lattices,gentry2008trapdoors,lyubashevsky2010ideal} (see, e.g.,~\cite{micciancio2009lattice} for a more detailed review).
Several works in the classical world have also shown the equivalence between the non-existence of one-way functions and the existence of efficient (average-case) learning algorithms for various problems~\cite{impagliazzo1990no,naor2006learning,naor2015bloom,hirahara2023learning}.

In contrast to this bountiful literature in the classical world, results investigating the intersection of quantum learning and cryptography have been significantly more scarce.
Some works have established computational hardness for learning classes of quantum states/unitaries from the existence of pseudorandom quantum states/unitaries~\cite{zhao2024learning,yang2023complexity,foxman2025random}.
Recently, researchers have begun to explore the converse: designing quantum cryptography from the \emph{average-case hardness of learning (AHL)} classes of quantum states~\cite{hiroka2024computational,fefferman2025hardness,hiroka2025hardness,niroula2026digital}\footnote{We remark that \cite{hiroka2025hardness} studies a different hardness-of-learning task, namely quantum distribution learning, where one wishes to learn classical bitstrings sampled from a distribution generated by a QPT algorithm.}.
Notably,~\cite{hiroka2024computational,fefferman2025hardness} prove that AHL for classes of pure states is equivalent to the existence of pure \emph{one-way state generators (OWSGs)}, i.e., efficient quantum algorithms that generate pure states that are easy to produce but computationally hard to invert.
However, this equivalence is unknown for the case of \emph{mixed states}, leading us to the central question of this work:
\begin{center}
    \emph{Is the average-case hardness of learning mixed states equivalent to mixed one-way state generators?}
\end{center}
This can be viewed as a more complete analogue of the equivalence between the average-case hardness of learning and the existence of one-way functions~\cite{impagliazzo1990no,naor2006learning,naor2015bloom,hirahara2023learning} in the classical setting.
Prior work~\cite{hiroka2024computational} provided initial progress on this question by showing that the existence of EFI\footnote{Here, an EFI pair is a pair of efficiently generatable quantum states that are statistically far but computationally indistinguishable~\cite{brakerski2022computational}.} implies AHL for classes of mixed states by using the equivalence of EFI with a variant of OWSGs, namely secretly-verifiable statistically-invertible OWSGs~\cite{morimae2022one}.
However, it is not clear if proving the converse is possible with this approach, as this variant of OWSGs requires an orthogonality property (statistical invertibility) which an arbitrary AHL instance does not satisfy in general.

In this work, we identify the appropriate cryptographic counterpart of mixed-state AHL: \emph{inefficiently verifiable one-way state generators (IV-OWSGs)}. Our main result shows that a class of mixed quantum states is hard to learn on average if and only if IV-OWSGs exist.

Our equivalence also transfers assumptions and applications between quantum learning and cryptography. Combining our main result with known relationships between IV-OWSGs and EFI pairs~\cite{malavolta2024exponential} shows that EFI implies mixed-state AHL, while a sufficiently strong, exponentially-hard AHL assumption implies EFI. Consequently, AHL in this parameter regime implies statistically binding quantum commitments and semi-honest oblivious transfer~\cite{brakerski2022computational}. In the other direction, any assumption known to imply IV-OWSGs, in particular, assumptions yielding EFI pairs, also gives rise to a class of mixed states that is hard to learn on average. Thus, AHL provides a learning-theoretic lens on cryptographic assumptions that may lie below classical one-way functions.

Finally, by combining existing results, it can be shown that IV-OWSGs exist relative to the SWAP oracle~\cite{goldin2025translating}, whereas OWSGs do not. The construction proceeds from statistically secure single-copy pseudorandom states in the common Haar random state model~\cite{chen2025power} and applies the lifting result of~\cite{goldin2025translating} to obtain the corresponding primitive in the SWAP model.
This oracle separation shows that our characterization cannot, in general, be upgraded from IV-OWSGs to efficiently-verifiable OWSGs by a relativizing argument.

We remark that this does not contradict the equivalence proved in the pure state case~\cite{hirahara2023learning,fefferman2025hardness} because verification can always be made efficient for pure states (see Lemma B.1 in~\cite{morimae2022one}).
Preliminary versions of our results appeared in~\cite{lewis2025computational}.

\subsection{Problem Definition and Results}
\label{sec:ahl-results}

The average-case hardness of learning (AHL) problem intuitively states that a class of quantum states is hard to learn on average if, given polynomially many copies of a state randomly sampled from the class, any efficient adversary can only learn a state close to it with at most some probability $\delta$. 
More formally, following~\cite{hiroka2024computational}, AHL is defined as follows\footnote{We have made some minor modifications to the definition from~\cite{hiroka2024computational} to align better with notation used in the learning theory literature, but this definition nonetheless captures the same premise.}:

\begin{definition}[Average-case hardness of learning (AHL); Definition 3.1 in~\cite{hiroka2024computational}]
  \label{def:ahl}
    Let $\mathcal{C}$ be a class of $n$-qubit quantum states $\rho_x$, indexed by $x \in \{0,1\}^n$, where each $\rho_x$ is generatable by a QPT algorithm.
    The class $\mathcal{C}$ is \emph{$(\epsilon, \delta)$-hard-to-learn} if there exists an efficiently sampleable distribution $\mathcal{D}_\lambda$ over $\{0,1\}^n$ such that for all $t(\lambda) = \mathrm{poly}(\lambda)$ and uniform QPT learning algorithms $\mathcal{A}$,
    \begin{equation}
        \Pr_{\substack{\rho_x \leftarrow \mathcal{D}_\lambda(\mathcal{C})\\y\leftarrow \mathcal{A}(\rho_x^{\otimes t(\lambda)})}}\left[\dtr(\rho_x, \rho_y) \leq \epsilon(\lambda)\right] \leq \delta(\lambda),
    \end{equation}
    where we use $\mathcal{D}_\lambda(\mathcal{C})$ to denote the distribution over $\mathcal{C}$ induced by $\mathcal{D}_\lambda$, which is over the labels of the states.
    When $\epsilon(\lambda), \delta(\lambda) = 1/\mathrm{poly}(\lambda)$, we simply say that $\mathcal{C}$ is \emph{hard on average}.
\end{definition}
\noindent In the above definition, we think of $\lambda$ as the security parameter, while $n = \mathrm{poly}(\lambda)$.

When $\epsilon(\lambda), \delta(\lambda)=1/\mathrm{poly}(\lambda)$, this assumption states that efficient algorithms can only learn the target state up to inverse-polynomial accuracy with inverse-polynomial probability.
One can also strengthen the assumption by choosing $\epsilon(\lambda) = 1/\mathrm{poly}(\lambda), \delta(\lambda) = 1/2^{cn}$ for some constant $c > 0$, which is still reasonable from our current understanding of tomography.
Note that one must have $c < 1$, as otherwise, the AHL assumption becomes trivial~\cite{fefferman2025hardness}.
We refer the reader to~\cite{fefferman2025hardness} for further discussion of noteworthy parameter regimes.

We also remark that this definition only considers \emph{proper} learning, i.e., learning algorithms are only allowed to output a hypothesis from the class $\mathcal{C}$.
This is consistent with previous work~\cite{hiroka2024computational,fefferman2025hardness}.
There is also a notion of improper learning, where the algorithm can approximate $\rho_x$ using any (efficient) description of a quantum state, instead of being restricted to only states in $\mathcal{C}$, but we do not consider improper learning in this work.

The other main object we consider is the \emph{one-way state generator}~\cite{morimae2022one}.
Informally, a one-way state generator consists of a set of algorithms ($\KeyGen, \StateGen, \Ver$) defined as follows:
\begin{itemize}
    \item $\KeyGen$ is a QPT algorithm that outputs a classical key $k$.
    \item $\StateGen(k)$ is a QPT algorithm that, on input key $k$, outputs a (potentially mixed) quantum state $\rho_k$.
    \item $\Ver(k', \rho_k)$ is a QPT algorithm that, on input $\rho_k$ and a bitstring $k'$, outputs $\top$ or $\bot$.
\end{itemize}
The security of a OWSG intuitively says that no QPT adversary can invert the output of $\StateGen$; in other words, given polynomially-many copies of $\rho_k$, it should be hard to recover $k$.
An \emph{inefficiently-verifiable} OWSG (IV-OWSG) is the same as a OWSG except the verification is not restricted to be QPT \cite{malavolta2024exponential,batra2024commitments}.
We refer the reader to \Cref{def:owsg,def:iv-owsg} for formal definitions.
Our main result is as follows.

\begin{theorem}[Equivalence between mixed AHL and IV-OWSGs]
    \label{thm:equiv}
    There exists a class of states that is average-case hard to learn if and only if there exist IV-OWSGs.
\end{theorem}

As an immediate corollary, because IV-OWSG is equivalent to EFI with an exponential loss in the reduction (see Theorems 4.3 and 5.2 in~\cite{malavolta2024exponential}), we obtain that AHL is equivalent to EFI, also up to an exponential loss in the reduction.

\begin{corollary}[Relationship between mixed AHL and EFI]
    If EFI pairs exist, then there exists a class of states that is average-case hard to learn.
    Moreover, if there exists a class of states that is average-case $(\epsilon,\delta)$-hard-to-learn, for $\epsilon(\lambda) = 1/\mathrm{poly}(\lambda)$ and $\delta(\lambda) = 2^{-0.75\lambda}$, then EFI pairs exist.
\end{corollary}

This corollary imports known applications of EFI into the learning setting. In particular, AHL with the stated exponential hardness implies statistically binding quantum commitments and semi-honest oblivious transfer~\cite{brakerski2022computational}. Conversely, EFI and any other assumption implying IV-OWSGs yield classes of mixed quantum states that are hard to learn on average.

More broadly, our results position AHL as a candidate foundational assumption for quantum cryptography in \emph{Microcrypt}, where useful quantum cryptographic primitives may exist even in the absence of quantum-secure one-way functions. Identifying natural hardness assumptions that suffice for cryptography but are potentially weaker than one-way functions is a central goal in quantum cryptography.
For example, recent proposals have drawn on the hardness of learning quantum states and the hardness of implementing non-collapsing measurements~\cite{fefferman2025hardness,morimae25collapsing}. Our equivalence places mixed-state AHL within this broader program and motivates studying whether AHL can hold without quantum-secure one-way functions.

In addition, we collect several existing results in the literature to show a separation between OWSG and IV-OWSG relative to the SWAP oracle~\cite{goldin2025translating} (see \Cref{def:oracle}).
This shows that a strengthening of \Cref{thm:equiv} to an equivalence between OWSGs and AHL is not possible.

\begin{theorem}
    Relative to the SWAP oracle, IV-OWSG exist while OWSG do not exist.
\end{theorem}

We briefly discuss the proof of \Cref{thm:equiv}.

\paragraph{AHL implies IV-OWSG.}
Consider an AHL instance (as in \Cref{def:ahl}) specified by a class $\mathcal{C} = \{\rho_x\}$ of quantum states, an efficiently sampleable distribution $\mathcal{D}_\lambda$, and parameters $\epsilon(\lambda),\delta(\lambda)$.
As a natural approach to construct an IV-OWSG, one may consider taking $\KeyGen$ to sample from $\mathcal{D}_\lambda$ and $\StateGen(x)$ to output the corresponding state $\rho_x$ from $\mathcal{C}$.
However, the correct choice of verification algorithm is not immediately clear:  the canonical choice of verification algorithm for pure states~\cite{morimae2022one} does not apply for mixed states; quantum state certification algorithms~\cite{buadescu2019quantum} require too many copies of the unknown state; the SWAP test for mixed states no longer reflects the trace distance between two states.
Our key observation is that one can use quantum hypothesis selection~\cite{buadescu2021improved} to instantiate the verification algorithm.
Algorithms for quantum hypothesis selection are computationally inefficient, resulting in an IV-OWSG.

\paragraph{OWSG implies AHL.}
Consider an IV-OWSG defined by three algorithms $(\KeyGen, \StateGen,$ $\Ver)$.
For this direction, one can consider the simple construction of AHL: take the distribution $\mathcal{D}_\lambda$ to sample from the same distribution as $\KeyGen$ and the class of states $\mathcal{C}$ to be the same set of states output by $\StateGen$.
Intuitively, the result of the theorem is reasonable: if one cannot invert a quantum state (i.e., find its key), then it should also be hard to learn a classical description of it, which is precisely specified by the key.
However, in AHL, the problem definition (\Cref{def:ahl}) allows the learning algorithm to recover keys of nearby states as well.
Meanwhile, a priori, the verification algorithm of the OWSG only accepts exactly the correct key.
To reconcile these definitional differences, previous work \cite{hiroka2024computational} instead considers OWSGs with an additional property of statistical invertibility: any states with different keys are far apart in trace distance (see \Cref{def:sv-si-owsg}).
This circumvents the issue by ensuring that the learning algorithm in AHL can only learn exactly the correct key, as any other key will not be within the allowed margin of error.
In contrast, our crucial observation is that the verification algorithm in fact also accepts keys corresponding to nearby states.

\subsection*{Organization}
In \Cref{sec:crypto}, we introduce basic cryptographic primitives that we refer to throughout the work.
In \Cref{sec:learning}, we state a guarantee from quantum learning theory which we find useful.
In \Cref{sec:ahl-to-owsg,sec:owsg-to-ahl}, we prove \Cref{thm:equiv}.
Specifically, in \Cref{sec:ahl-to-owsg}, we prove that AHL implies IV-OWSG, and in \Cref{sec:owsg-to-ahl}, we prove that IV-OWSG implies AHL.
In \Cref{sec:separation}, we compile results from the literature to show an oracle separation between IV-OWSG and OWSG.

\subsection*{Acknowledgments}
The authors thank Damiano Abram, Mohammed Barhoush, Taiga Hiroka, Dominik Leichtle, Thomas Vidick, and Chirag Wadhwa,  for helpful discussions at various stages of this project.
A.C. acknowledges support from the National Science Foundation grant CCF-1813814, from the AFOSR under Award Number FA9550-20-1-0108 and from the Quantum Advantage Pathfinder project.
L.L. was supported by a Marshall Scholarship and a U.S. Department of
Energy, Office of Science, Office of Advanced Scientific Computing Research, Department of Energy Computational Science Graduate Fellowship under Award Number DE-SC0026073.

This report was prepared as an account of work sponsored by an agency of the United States Government. Neither the United States Government nor any agency thereof, nor any of their employees, makes any warranty, express or implied, or assumes any legal liability or responsibility for the accuracy, completeness, or usefulness of any information, apparatus, product, or process disclosed, or represents that its use would not infringe privately owned rights. Reference herein to any specific commercial product, process, or service by trade name, trademark, manufacturer, or otherwise does not necessarily constitute or imply its endorsement, recommendation, or favoring by the United States Government or any agency thereof. The views and opinions of authors expressed herein do not necessarily state or reflect those of the United States Government or any agency thereof.

\subsection*{AI Use Disclosure}
All ideas, results, and writing were generated by the authors without AI use.
ChatGPT 5.6 Pro was used to check the correctness of the proofs after they were already completed by the authors.

\section{Preliminaries}
\subsection{Cryptographic Tools and Definitions}
\label{sec:crypto}

In this section, we review some ideas from cryptography, which we will use throughout the paper.
First, we define a one-way function, which is the most fundamental primitive in classical cryptography.
Throughout, we let $\lambda$ denote the security parameter and $n = \mathrm{poly}(\lambda)$.
Hereafter, by a negligible function we mean a function that decays faster than any polynomial.

\begin{definition}[One-Way Function (OWF)]
  Let $\mathcal{F} = \{f_\lambda\}_{\lambda \in \mathbb{N}}$ be a family of efficiently-computable functions $f_\lambda: \{0,1\}^\lambda \to \{0,1\}^{\ell(\lambda)}$.
  $\mathcal{F}$ is a \emph{one-way function} if for every polynomial-time probabilistic algorithm $\mathcal{A}$, there exists a negligible function $\negl(\cdot)$ such that for every security parameter $\lambda \in \mathbb{N}$
  \begin{equation}
    \Pr_{\substack{x \leftarrow \{0,1\}^\lambda\\y\leftarrow f_\lambda(x)}}[\mathcal{A}(1^\lambda, y) \in f_\lambda^{-1}(y)] \leq \negl(\lambda).
  \end{equation}
\end{definition}
In other words, a OWF is a function which is easy to compute but hard to invert for any polynomial-time algorithm.
In this work, we mainly focus on the \emph{quantum version} of one-wayness.
In particular, one can consider quantum objects that behave similarly to a OWF, e.g., a keyed quantum state for which it is hard to efficiently recover the key.
This is exactly the notion of a one-way state generator.
In the following definition, we parameterize the security.

\begin{definition}[$\eta$-Secure One-Way State Generators ($\eta$-Secure OWSGs); Definition 3.1 in~\cite{morimae2022one}]
\label{def:owsg}
  An \emph{$\eta$-secure one-way state generator (OWSG)} is a set of algorithms ($\KeyGen, \StateGen, \Ver$) such that
  \begin{itemize}
    \item $\KeyGen(1^\lambda)$ is a QPT algorithm tha   t, on input the security parameter $\lambda$, outputs a classical key $k \in \{0,1\}^n$.
    \item $\StateGen(k)$ is a QPT algorithm that, on input key $k \in \{0,1\}^n$, outputs a (potentially mixed) quantum state $\rho_k$.
    \item $\Ver(k', \rho_k)$ is a QPT algorithm that, on input $\rho_k$ and a bitstring $k'$, outputs $\top$ or $\bot$.
  \end{itemize}
  Moreover, these algorithms satisfy the following properties
  \begin{itemize}
    \item (Correctness) For any key $k \in \{0,1\}^n$, then
    \begin{equation}
      \Pr_{\substack{\rho_k \leftarrow \StateGen(k)}}[\top \leftarrow \Ver(k, \rho_k)] \geq 1- \negl(\lambda).
    \end{equation}
    \item (Security) For any uniform QPT adversary $\mathcal{A}$ and polynomial $t$,
    \begin{equation}
      \Pr_{\substack{k \leftarrow \KeyGen(1^\lambda)\\\rho_k\leftarrow \StateGen(k)}}[\top \leftarrow \Ver(\mathcal{A}(\rho_k^{\otimes t(\lambda)}), \rho_k)] \leq \eta.
    \end{equation}
  \end{itemize}
  If $\StateGen$ only outputs pure states, we call these \emph{pure OWSGs}.
  Moreover, when $\eta = \negl(\lambda)$, we simply call these OWSGs.
\end{definition}

It is instructive to note that for pure OWSGs, without loss of generality, one may consider $\Ver$ as the following algorithm~\cite{morimae2022one,morimae2022quantum}: on input $k'$ and $\ket{\phi_k}$, measure $\phi_k$ with the projective measurement $\{\ketbra{\phi_{k'}}, I - \ketbra{\phi_{k'}}\}$.
If the result is $\ketbra{\phi_{k'}}$, output $\top$; otherwise, output $\bot$.

We also note that in some definitions of OWSGs, the correctness is only required to be probabilistic over the choice of key.
We modify this slightly to hold for all keys instead, following~\cite{fefferman2025hardness}.
There are also several other variations of OWSGs.
First, we consider inefficiently-verifiable OWSGs (IV-OWSGs).

\begin{definition}[$\eta$-Secure Inefficiently-verifiable OWSGs (IV-OWSGs); Definition 15 in~\cite{batra2024commitments} or Definition 3.2 in~\cite{malavolta2024exponential}]
\label{def:iv-owsg}
An \emph{inefficiently-verifiable one-way state generator (IV-OWSG)} is defined in the same way as a OWSG, except that $\Ver$ is allowed to be inefficient.
\end{definition}

We also introduce secretly-verifiable OWSGs (SV-OWSGs)~\cite{morimae2022one}.
Here, one can think of the verification algorithm $\Ver$ as being given two classical keys $k,k'$ instead of a classical key $k'$ and a copy of an unknown quantum state $\rho_k$ as in \Cref{def:owsg}.
Then, without loss of generality, one can replace $\Ver$ with simply checking if $k = k'$.

\begin{definition}[Secretly-verifiable OWSGs (SV-OWSGs); Definition 7.1 in~\cite{morimae2022one}]
\label{def:sv-owsg}
A \emph{secretly-verifiable one-way state generator (SV-OWSG)} is a set of algorithms $(\KeyGen, \StateGen)$, where $\KeyGen, \StateGen$ are defined in the same way as for OWSGs. The only difference is that security is defined as
\begin{itemize}
  \item (Security) For any uniform QPT adversary $\mathcal{A}$ and polynomial $t$,
  \begin{equation}
    \Pr_{\substack{k \leftarrow \KeyGen(1^\lambda)\\\rho_k \leftarrow \StateGen(k)}}[k \leftarrow \mathcal{A}(\rho_k^{\otimes t(\lambda)})] \leq \negl(\lambda).
  \end{equation}
\end{itemize}
\end{definition}

Moreover, there is another variant which is similar to SV-OWSGs but additionally requires that all of the states generated are far apart in trace distance.

\begin{definition}[Secretly-verifiable statistically-invertible OWSGs (SV-SI-OWSGs); Definition 7.3 in~\cite{morimae2022one}]
\label{def:sv-si-owsg}
A \emph{secretly-verifiable and statistically-invertible one-way state generator (SV-SI-OWSG)} is a set of algorithms $(\KeyGen, \StateGen)$ such that
\begin{itemize}
  \item $\KeyGen(1^\lambda)$ is a QPT algorithm that, on input the security parameter $\lambda$, outputs a classical key $k\in \{0,1\}^n$.
  \item $\StateGen(k)$ is a QPT algorithm that, on input key $k \in \{0,1\}^n$, outputs a (potentially mixed) quantum state $\rho_k$.
\end{itemize}
Moreover, these algorithms satisfy the following properties
\begin{itemize}
  \item (Statistical invertibility) For any $k, k'$ with $k \neq k'$, then
  \begin{equation}
    \dtr(\rho_k, \rho_{k'}) \geq 1-\negl(\lambda).
  \end{equation}
  \item (Computational non-invertibility) For any uniform QPT adversary $\mathcal{A}$ and polynomial $t$,
  \begin{equation}
    \Pr_{\substack{k \leftarrow \KeyGen(1^\lambda)\\\rho_k \leftarrow \StateGen(k)}}[k \leftarrow \mathcal{A}(\rho_k^{\otimes t(\lambda)})] \leq \negl(\lambda).
  \end{equation}
\end{itemize}
\end{definition}

While this may seem like a strange requirement at first, SV-SI-OWSGs are interesting because they are equivalent to EFI, a primitive believed to be a minimal assumption for quantum cryptography.

\begin{definition}[EFI; Definition 3.1 in~\cite{brakerski2022computational}]
\label{def:efi}
An \emph{EFI pair} is a uniform QPT algorithm \textsf{EFI} such that
\begin{itemize}
  \item (Efficient generation) $\EFI$ is a uniform QPT algorithm that on input $(1^\lambda, b)$ for a security parameter $\lambda \in \mathbb{N}$ and a bit $b \in \{0,1\}$, outputs a (potentially mixed) quantum state $\rho_b$. That is, $\EFI(1^\lambda, b) = \rho_{b}$.
  \item (Statistically far) 
  \begin{equation}
    \dtr(\rho_{0}, \rho_{1}) \geq 1 - \negl(\lambda).
  \end{equation}
  \item (Computational indistinguishability) For any uniform QPT adversary $\mathcal{A}$,
  \begin{equation}
    |\Pr[1 \leftarrow \mathcal{A}(\rho_{0})] - \Pr[1 \leftarrow \mathcal{A}(\rho_{1})]| \leq \negl(\lambda).
  \end{equation}
\end{itemize}
\end{definition}

\begin{theorem}[Theorem 7.7 in~\cite{morimae2022one}]
  SV-SI-OWSGs exist if and only if EFI pairs exist.
\end{theorem}

\subsection{Quantum Learning Theory}
\label{sec:learning}
First, recall the definition of trace distance between two quantum states.

\begin{definition}[Trace distance]
    Let $\rho, \sigma$ be two quantum states.
    The \emph{trace distance} between $\rho$ and $\sigma$ is
    \begin{equation}
        \dtr(\rho, \sigma) \triangleq \frac{1}{2}\norm{\rho - \sigma}_1,
    \end{equation}
    where $\norm{M}_1 = \tr(\sqrt{M^\dagger M})$ denotes the trace norm.
\end{definition}

We also find the following guarantee from quantum learning theory useful.
It provides a sample complexity upper bound for the task of \emph{hypothesis selection}, i.e., finding the best approximation to an unknown state from a class of hypothesis states.

\begin{theorem}[Theorem 1.5 in~\cite{buadescu2021improved}]
    \label{thm:hypothesis-selection}
    There is a quantum algorithm that, given $m$ fixed hypothesis states $\mathcal{C} = \{\sigma_1,\dots,\sigma_m\}$, where $\sigma_i \in \mathbb{C}^{d\times d}$, parameters $0 < \epsilon,\delta < 1/2$, and access to unentangled copies of a state $\rho \in \mathbb{C}^{d\times d}$, where $\rho \in \mathcal{C}$, uses
    \begin{equation}
        N(\epsilon, \delta) = \min\left\{\frac{(\log^2 m + L_1)(\log d)}{\epsilon^4} \cdot \mathcal{O}(L_1), \frac{\log^3m + \log(L_2/\delta) \cdot \log m}{\epsilon^2} \cdot \mathcal{O}(L_2 \log(L_2/\delta)\right\}
    \end{equation}
    copies of $\rho$, where $L_1 = \log\left(\frac{\log d}{\delta\epsilon}\right)$ and $L_2 = \log(1/\epsilon)$ and has the following guarantee: with probability at least $1-\delta$, it outputs $k$ such that
    \begin{equation}
        \dtr(\rho, \sigma_k) \leq \epsilon.
    \end{equation}
\end{theorem}

\section{AHL implies IV-OWSG}
\label{sec:ahl-to-owsg}

In this section, we prove that AHL implies IV-OWSG.
For pure states, it is easy to construct IV-OWSGs from AHL because there is a canonical choice of verification algorithm (simply measure the projector).
However, the main difficulty for mixed states is the absence of such a verification algorithm, where, e.g., quantum state certification algorithms or the SWAP test require too many copies of the unknown state.
Our key observation is that one can use quantum hypothesis selection (\Cref{thm:hypothesis-selection}) to instantiate the verification algorithm for mixed states.
However, hypothesis selection is computationally inefficient, resulting in an IV-OWSG.

\begin{prop}[AHL $\Rightarrow$ IV-OWSG]
\label{prop:ahl-to-iv-owsg}
If there exists a class of states that is average-case hard to learn (as in \Cref{def:ahl}), then there exists an IV-OWSG.
\end{prop}

\begin{proof}
Suppose there exists an AHL instance specified by a class $\mathcal{C}$ of $n$-qubit quantum states indexed by bitstrings in $\mathcal{X} \subseteq \{0,1\}^n$, an efficiently sampleable distribution $\mathcal{D}_\lambda$ over $\mathcal{X}$, and $\epsilon(\lambda),\delta(\lambda) = 1/\mathrm{poly}(\lambda)$.
By definition of AHL, for all polynomials $t(\lambda) = \mathrm{poly}(\lambda)$ and all uniform QPT learning algorithms $\mathcal{A}$, then
\begin{equation}
    \Pr_{\substack{\rho_x \leftarrow \mathcal{D}_\lambda(\mathcal{C})\\y \leftarrow \mathcal{A}(\rho_x^{\otimes t(\lambda)})}}[\dtr(\rho_x, \rho_y) \leq \epsilon(\lambda)] \leq \delta(\lambda).
\end{equation}
Here, we use $\mathcal{D}_\lambda(\mathcal{C})$ to denote the distribution over $\mathcal{C}$ induced by $\mathcal{D}_\lambda$, which is over the labels
of the states instead. 
We construct an IV-OWSG as follows:
\begin{itemize}
    \item $\KeyGen(1^\lambda)$ outputs $x \in \{0,1\}^n$ sampled from $\mathcal{D}_\lambda$.
    \item $\StateGen(x)$ outputs $\rho_x^{\otimes N}$, where $\rho_x \in \mathcal{C}$ and $N = N(\epsilon(\lambda)/2, 1/2^n)$, where $N$ is defined in \Cref{thm:hypothesis-selection}.
    \item $\Ver(y, \rho_x^{\otimes N})$ first runs the hypothesis selection algorithm from \Cref{thm:hypothesis-selection} on $\rho_x^{\otimes N}$ and the class $\mathcal{C}$ to obtain an output $x'$. Then, it outputs $\top$ iff $\dtr(\rho_y, \rho_{x'}) \leq \epsilon(\lambda)/2$.
\end{itemize}
Note that $N = N(\epsilon(\lambda)/2, 1/2^n) = \mathrm{poly}(n)$ because $N(\epsilon, \delta) = \mathrm{poly}(1/\epsilon, \log(1/\delta), \log(m))$, and the size of the hypothesis class is $m = |\mathcal{C}| \leq 2^n$.
Thus, $\StateGen$ is a QPT algorithm because it only needs to output a $\mathrm{poly}(n)$ copies of an efficiently generatable state (since all states in $\mathcal{C}$ can be generated by a QPT algorithm).

We also remark that this is only an IV-OWSG because the hypothesis selection algorithm from~\cite{buadescu2021improved} is computationally inefficient.
Also, because the verification is allowed to be inefficient and $y,x'$ are known along with a classical description of the class $\mathcal{C}$, then the trace distance calculation can be performed.

First, we show correctness, i.e., for any $x \leftarrow \KeyGen(1^\lambda)$, then
\begin{equation}
    \Pr_{\rho_x^{\otimes N} \leftarrow \StateGen(x)}[\top \leftarrow \Ver(x, \rho_x^{\otimes N})] \geq 1 - \negl(\lambda).
\end{equation}
By the guarantees of hypothesis selection with our choice of parameters (\Cref{thm:hypothesis-selection}), then running hypothesis selection on $\rho_x^{\otimes N}$ and the class $\mathcal{C}$ outputs some $x'$ such that $\dtr(\rho_x, \rho_{x'}) \leq \epsilon(\lambda)/2$ with probability at least $1-1/2^n$.
Then, the verification algorithm simply checks if $\dtr(\rho_x, \rho_{x'}) \leq \epsilon(\lambda)/2$, which clearly holds as long as hypothesis selection is successful.
Thus, it is clear that we have
\begin{equation}
    \Pr_{\rho_x^{\otimes N} \leftarrow \StateGen(x)}[\top \leftarrow \Ver(x, \rho_x^{\otimes N})] \geq 1 - \frac{1}{2^n}.
\end{equation}

It remains to show security.
We want to show that for any uniform QPT adversary $\mathcal{A}$ and any polynomial $t$,
\begin{equation}
	\Pr_{\substack{x\leftarrow \KeyGen(1^\lambda)\\\rho_x^{\otimes N}\leftarrow \StateGen(x)}}[\top \leftarrow \Ver(\mathcal{A}(\rho_x^{\otimes (N \cdot t(\lambda))}), \rho_x^{\otimes N})] \leq \negl(\lambda).
\end{equation}
Suppose for the sake of contradiction that there exists a uniform QPT algorithm $\mathcal{B}$ and a polynomial $t$ such that
\begin{equation}
	\Pr_{\substack{x\leftarrow \KeyGen(1^\lambda)\\\rho_x^{\otimes N}\leftarrow \StateGen(x)}}[\top \leftarrow \Ver(\mathcal{B}(\rho_x^{\otimes (N \cdot t(\lambda))}), \rho_x^{\otimes N})] > \delta(\lambda) + \frac{1}{2^n}.
\end{equation}
To arrive at a contradiction, we show that $\mathcal{B}$ can also learn $\mathcal{C}$, contradicting AHL.
Define the set
\begin{equation}
	S_{\mathrm{good}} \triangleq \left\{(x,y) \in \mathcal{X}^2 : \dtr(\rho_x, \rho_y) \leq \epsilon(\lambda), \rho_x \leftarrow \mathcal{C}(x), \rho_y \leftarrow \mathcal{C}(y)\right\}.
\end{equation}
Recall here that $\mathcal{X}$ is the set of bitstring labels for states in the class $\mathcal{C}$.
Moreover, we use $\rho_x \leftarrow \mathcal{C}(x)$ to denote the state $\rho_x \in \mathcal{C}$ corresponding to the bitstring label $x$.
Then, we have
\begin{align}
    &\delta(\lambda) + \frac{1}{2^n}\\
    &< \Pr_{\substack{x\leftarrow \KeyGen(1^\lambda)\\\rho_x^{\otimes N}\leftarrow \StateGen(x)}}[\top \leftarrow \Ver(\mathcal{B}(\rho_x^{\otimes (N \cdot t(\lambda))}), \rho_x^{\otimes N})]\\
    &= \sum_{x,y} \Pr[x \leftarrow \KeyGen(1^\lambda)] \Pr_{\rho_x^{\otimes N} \leftarrow \StateGen(x)}[y \leftarrow \mathcal{B}(\rho_x^{\otimes (N \cdot t(\lambda))})]\Pr_{\rho_x^{\otimes N} \leftarrow \StateGen(x)}[\top \leftarrow \Ver(y, \rho_x^{\otimes N})]\\
    &= \sum_{x,y \in S_{\mathrm{good}}} \Pr[x \leftarrow \mathcal{D}_\lambda] \Pr_{\rho_x \leftarrow \mathcal{C}(x)}[y \leftarrow \mathcal{B}(\rho_x^{\otimes (N \cdot t(\lambda))})]\Pr_{\rho_x \leftarrow \mathcal{C}(x)}[\top \leftarrow \Ver(y, \rho_x^{\otimes N})]\\
    &+ \sum_{x,y \not\in S_{\mathrm{good}}} \Pr[x \leftarrow \mathcal{D}_\lambda] \Pr_{\rho_x \leftarrow \mathcal{C}(x)}[y \leftarrow \mathcal{B}(\rho_x^{\otimes (N \cdot  t(\lambda))})]\Pr_{\rho_x \leftarrow \mathcal{C}(x)}[\top \leftarrow \Ver(y, \rho_x^{\otimes N})]\\
    &\leq \sum_{x,y \in S_{\mathrm{good}}} \Pr[x \leftarrow \mathcal{D}_\lambda] \Pr_{\rho_x \leftarrow \mathcal{C}(x)}[y \leftarrow \mathcal{B}(\rho_x^{\otimes (N \cdot t(\lambda))})] + \frac{1}{2^n}\sum_{x,y \notin S_{\mathrm{good}}} \Pr[x \leftarrow \mathcal{D}_\lambda] \Pr_{\rho_x \leftarrow \mathcal{C}(x)}[y \leftarrow \mathcal{B}(\rho_x^{\otimes (N \cdot t(\lambda))})]\\
    &= \sum_{x,y \in S_{\mathrm{good}}} \Pr[x \leftarrow \mathcal{D}_\lambda] \Pr_{\rho_x \leftarrow \mathcal{C}(x)}[y \leftarrow \mathcal{B}(\rho_x^{\otimes (N \cdot t(\lambda))})] + \frac{1}{2^n} \Pr_{\substack{\rho_x \leftarrow \mathcal{D}_\lambda(\mathcal{C})\\y\leftarrow \mathcal{B}(\rho_x^{\otimes (N \cdot t(\lambda))})}}[\dtr(\rho_x, \rho_y) > \epsilon(\lambda)]\\
    &\leq \sum_{x,y \in S_{\mathrm{good}}} \Pr[x \leftarrow \mathcal{D}_\lambda] \Pr_{\rho_x \leftarrow \mathcal{C}(x)}[y \leftarrow \mathcal{B}(\rho_x^{\otimes (N \cdot t(\lambda))})] + \frac{1}{2^n},\label{eq:good-prob-lower}
\end{align}
where in the third line, we use the definition of our construction of IV-OWSGs.
In the fifth line, we use the definition of our verification procedure, which we claim outputs $\top$ for $(x,y) \notin S_{\mathrm{good}}$ with probability at most $1/2^n$.
This is because by \Cref{thm:hypothesis-selection} and our choice of parameters, with probability at least $1-1/2^n$, hypothesis selection outputs some $x'$ such that $\dtr(\rho_x, \rho_{x'}) \leq \epsilon(\lambda)/2$.
Moreover, for $(x,y)\notin S_{\mathrm{good}}$, then $\dtr(\rho_x, \rho_y) > \epsilon(\lambda)$.
Then, by reverse triangle inequality,
\begin{equation}
    \dtr(\rho_y, \rho_{x'}) \geq |\dtr(\rho_y, \rho_x) - \dtr(\rho_x, \rho_{x'})| > \epsilon(\lambda)/2.
\end{equation}
Hence, because $\Ver(y, \rho_x^{\otimes N})$ outputs $\top$ iff $\dtr(\rho_y, \rho_{x'}) \leq \epsilon(\lambda)/2$, then in this case, the verification outputs $\bot$ unless hypothesis selection fails, which occurs with probability at most $1/2^n$.

Finally, writing the AHL security probability, we have, for a polynomial $t'(\lambda) = N \cdot t(\lambda)$, then
\begin{align}
	\Pr_{\substack{\rho_x \leftarrow \mathcal{D}_\lambda(\mathcal{C})\\y \leftarrow \mathcal{B}(\rho_x^{\otimes t'(\lambda)})}}\left[\dtr(\rho_x, \rho_y) \leq \epsilon(\lambda)\right] &= \sum_{x,y} \Pr[x \leftarrow \mathcal{D}_\lambda] \Pr_{\rho_x \leftarrow \mathcal{C}(x)}[y \leftarrow \mathcal{B}(\rho_x^{\otimes t'(\lambda)})] \Pr_{\substack{\rho_x \leftarrow \mathcal{C}(x)\\\rho_y\leftarrow \mathcal{C}(y)}}\left[\dtr(\rho_x, \rho_y) \leq \epsilon(\lambda)\right]\\
	&= \sum_{x,y \in S_{\mathrm{good}}} \Pr[x \leftarrow \mathcal{D}_\lambda] \Pr_{\rho_x \leftarrow \mathcal{C}(x)}[y\leftarrow \mathcal{B}(\rho_x^{\otimes t'(\lambda)})]\\
	&> \delta(\lambda),
\end{align}
which is a contradiction.
In the last line, we use \Cref{eq:good-prob-lower}.
Thus, we have proven that for all uniform QPT algorithms $\mathcal{A}$ and polynomials $t$, then
\begin{equation}
	\Pr_{\substack{x \leftarrow \KeyGen(1^\lambda)\\\rho_x^{\otimes N} \leftarrow \StateGen(x)}}[\top \leftarrow \Ver(\mathcal{A}(\rho_x^{\otimes (N \cdot t(\lambda))}), \rho_x^{\otimes N})] \leq \delta(\lambda) + \frac{1}{2^n}
\end{equation}
for an inverse-polynomial $\delta$.
This can be amplified to negligible security by Theorem 3.7 of~\cite{morimae2022one}.
\end{proof}

\section{IV-OWSG implies AHL}
\label{sec:owsg-to-ahl}

In this section we show that IV-OWSG implies AHL.
Intuitively, this claim makes sense: if one cannot invert the state (i.e., find its key), then it should also be hard to learn a classical description of it, which is precisely specified by the key.
Despite this intuition, we were not able to find a proof of this in the literature, and \cite{hiroka2024computational} only achieves a relationship between OWSGs and AHL in the mixed case for the different notion of SV-SI-OWSGs (see \Cref{def:sv-si-owsg}).
Their proof also uses the statistically-invertible property of these OWSGs crucially.
Instead, we provide a simple proof showing that IV-OWSG implies AHL, even in the mixed state case.

Before presenting the proof of our result, we note that an initial, weaker result can be obtained easily by stitching together results from the literature.

\begin{prop}
    Let $D > 0$ be a constant and let $\eta(\lambda) \geq (0.5 + D)\lambda$.
    If $2^{-\eta(\lambda)}$-secure IV-OWSGs exist, then AHL exists.
\end{prop}

\begin{proof}
    This follows by combining Theorem 5.2 from~\cite{malavolta2024exponential} and Theorem 5.3 in~\cite{hiroka2024computational}.
    Namely, Theorem 5.2 in~\cite{malavolta2024exponential} proves that exponentially-secure IV-OWSGs imply EFI.
    Then, Theorem 5.3 in~\cite{hiroka2024computational} shows that EFI implies AHL for mixed states.
\end{proof}

This proves that IV-OWSG implies AHL, although with an exponential loss in the reduction.
In contrast, we prove that standard IV-OWSGs are sufficient to construct AHL, with no loss in the reduction.
The key intuition is that the $\Ver$ algorithm of the OWSG also accepts keys corresponding to states that are close in trace distance to the target state, just like AHL.

\begin{prop}[IV-OWSG $\Rightarrow$ AHL]
    If IV-OWSGs exist, then there exists a class of states that is average-case hard to learn (as in \Cref{def:ahl}).
\end{prop}

\begin{proof}
Suppose there exists an IV-OWSG specified by algorithms $(\KeyGen, \StateGen, \Ver)$.
We construct a class of states that is hard to learn on average from this.
This requires constructing a distribution $\mathcal{D}_\lambda$ and a class of states $\mathcal{C}$.
We simply define $\mathcal{D}_\lambda$ as sampling from the same distribution as $\KeyGen(1^\lambda)$ and $\mathcal{C}$ as the same set of states output by $\StateGen$.
We want to show that this class is indeed an AHL instance, i.e., there exist $\epsilon(\lambda),\delta(\lambda) = 1/\mathrm{poly}(\lambda)$ such that for all polynomials $t$ and uniform QPT algorithms $\mathcal{A}$,
\begin{equation}
  \Pr_{\substack{\rho_x \leftarrow \mathcal{D}_\lambda(\mathcal{C})\\y \leftarrow \mathcal{A}( \rho_x^{\otimes t(\lambda)})}}\left[\dtr(\rho_x, \rho_y)\leq \epsilon(\lambda)\right] \leq \delta(\lambda).
\end{equation}
Suppose for the sake of contradiction that for all inverse polynomial $\epsilon,\delta$, there exists a polynomial $t$ and QPT learner $\mathcal{B}$ such that
\begin{equation}
  \Pr_{\substack{\rho_x \leftarrow \mathcal{D}_\lambda(\mathcal{C})\\y \leftarrow \mathcal{B}( \rho_x^{\otimes t(\lambda)})}}\left[\dtr(\rho_x, \rho_y)\leq \epsilon(\lambda)\right] > \delta(\lambda).
\end{equation}
To arrive at a contradiction, we will show that $\mathcal{B}$ can also break the security of the OWSG.
Define the set
\begin{equation}
  S_{\mathrm{good}} \triangleq \{(x,y) \in \mathcal{X}^2 : \dtr(\rho_x, \rho_y) \leq \epsilon(\lambda), \; \rho_x \leftarrow \mathcal{C}(x), \rho_y \leftarrow \mathcal{C}(y)\}.
\end{equation}
Notice that by our assumption on the adversary $\mathcal{B}$ we have
\begin{align}
  &\Pr_{\substack{\rho_x \leftarrow \mathcal{D}_\lambda(\mathcal{C})\\y \leftarrow \mathcal{B}(\rho_x^{\otimes t(\lambda)})}}[\dtr(\rho_x, \rho_y) \leq \epsilon(\lambda)]\\
    &= \sum_{(x,y) \in \mathcal{X}^2}\Pr[x \leftarrow \mathcal{D}_\lambda] \Pr_{\substack{\rho_x \leftarrow \mathcal{C}(x)}}[y \leftarrow \mathcal{B}(\rho_x^{\otimes t(\lambda)})]\Pr_{\substack{\rho_x \leftarrow \mathcal{C}(x)\\\rho_y \leftarrow \mathcal{C}(y)}}[\dtr(\rho_x, \rho_y)\leq \epsilon(\lambda)]\\
    &= \sum_{(x,y) \in S_{\mathrm{good}}} \Pr[x \leftarrow \mathcal{D}_\lambda] \Pr_{\substack{\rho_x \leftarrow \mathcal{C}(x)}}[y \leftarrow \mathcal{B}(\rho_x^{\otimes t(\lambda)})]\\
    &> \delta(\lambda).\label{eq:contra3}
\end{align}
Now, we can lower bound the probability of verification passing for the OWSG using $\mathcal{B}$'s output:
\begin{align}
  &\Pr_{\substack{x \leftarrow \KeyGen(1^\lambda)\\\rho_x \leftarrow \StateGen(x)}}[\top \leftarrow \Ver(\mathcal{B}(\rho_x^{\otimes t(\lambda)}), \rho_x)]\\
    &= \sum_{(x,y) \in \mathcal{X}^2} \Pr[x \leftarrow \KeyGen(1^\lambda)] \Pr_{\substack{\rho_x \leftarrow \StateGen(x)}}[y \leftarrow \mathcal{B}(\rho_x^{\otimes t(\lambda)})]\Pr_{\rho_x \leftarrow \StateGen(x)}[\top \leftarrow \Ver(y, \rho_x)]\\
    &\geq \sum_{(x,y) \in S_{\mathrm{good}}} \Pr[x \leftarrow \KeyGen(1^\lambda)] \Pr_{\substack{\rho_x \leftarrow \StateGen(x)}}[y \leftarrow \mathcal{B}(\rho_x^{\otimes t(\lambda)})]\Pr_{\rho_x \leftarrow \StateGen(x)}[\top \leftarrow \Ver(y, \rho_x)]\\
    &\begin{aligned}
        \geq \sum_{(x,y) \in S_{\mathrm{good}}} \Pr[x \leftarrow \KeyGen(1^\lambda)] \Pr_{\substack{\rho_x \leftarrow \StateGen(x)}}&[y \leftarrow \mathcal{B}(\rho_x^{\otimes t(\lambda)})]\\
        &\cdot \left(\Pr_{\rho_y \leftarrow \StateGen(y)}[\top \leftarrow \Ver(y, \rho_y)] - \epsilon(\lambda)\right)
    \end{aligned}\\
    &> \delta(\lambda)\left(1 - \negl(\lambda) - \epsilon(\lambda)\right).
\end{align}
In the fourth line, we use that, by definition of $S_{\mathrm{good}}$, then for all $(x,y) \in S_{\mathrm{good}}$, $\dtr(\rho_x, \rho_y) \leq \epsilon(\lambda)$.
Thus, $|\tr(\Pi\rho_x) - \tr(\Pi\rho_y)| \leq \epsilon(\lambda)$ for any POVM $\Pi$.
In general, the verification procedure $\Ver(y, \rho_x)$ performs a (computationally inefficient) POVM on $\rho_x$, so the probability of it outputting $\top$ is given by $\tr(\Pi \rho_x)$.
In the last inequality, we used that $\Pr[\top \leftarrow \Ver(x,\rho_x)] \geq 1 -\negl(\lambda)$ for all $x$ (see our definition of OWSGs in \Cref{def:owsg}) and \Cref{eq:contra3} with our definition of the AHL instance.
Because $\delta, \epsilon$ are inverse polynomial, this is a contradiction to the security of the IV-OWSG.
\end{proof}

\section{Separation Between IV-OWSG and OWSG}
\label{sec:separation}

In this section, we stitch together several results in the literature to show an oracle separation between IV-OWSG and OWSG.
For our purposes, it is important that these results hold against uniform adversaries.
This separation holds relative to the SWAP oracle~\cite{goldin2025translating}, which we define as follows.

\begin{definition}[Oracle models; Definition 3.1 of~\cite{goldin2025translating}]
    \label{def:oracle}
    Let $\mathcal{D} = \{\mathcal{D}_n\}_{n \in \mathbb{N}}$ be a family of distributions over pure states on $n$-qubits, which are elements of $\mathbb{C}^{2^n}$.
    Let $\ket{\phi_n}$ be a state sampled from $\mathcal{D}$ during initialization, before any party receives oracle access.
    Consider the space $\mathbb{C}^{2^n+1}$ which will be the space spanned by $\mathbb{C}^{2^n}$ of $\mathcal{D}_n$ and an orthogonal basis vector $\ket{0}$.
    \begin{itemize}
        \item Define the \emph{common Haar random state model} over $\mathcal{D}$ to be the oracle $\mathrm{CHRS}_n: \mathbb{C} \to \mathbb{C}^{2^n}$ which acts as
        \begin{equation}
            \ket{0} \mapsto \ket{\phi_n}.
        \end{equation}
        \item Define the \emph{swap structured state model (SWAP)} over $\mathcal{D}$ to be the oracle $\mathrm{SWAP}_n: \mathbb{C}^{2^n+1} \to \mathbb{C}^{2^n+1}$ which acts as
        \begin{equation}
            \mathrm{SWAP}_n = I - \ketbra{0} - \ketbra{\phi_n} + \ketbra{0}{\phi_n} + \ketbra{1}{0}.
        \end{equation}
    \end{itemize}
\end{definition}

First, we state a result from~\cite{goldin2025translating}.

\begin{theorem}[Corollary C.7 of~\cite{goldin2025translating}]
    There does not exist a OWSG relative to SWAP.
\end{theorem}

Now, we show the following theorem, which establishes the oracle separation.

\begin{theorem}
    IV-OWSG exist relative to SWAP.
\end{theorem}

To show this, we state two results, which, when combined, give us the above statement.
Notably, all of these results hold for uniform adversaries.

\begin{theorem}[Theorem 4.1 of~\cite{chen2025power}]
    Statistically-secure 1PRS exist relative to SWAP.
\end{theorem}

\begin{proof}
    By Theorem 4.1 of~\cite{chen2025power}, statistically-secure 1PRS exist in the common Haar random state model.
    By applying Corollary 8.4 of~\cite{goldin2025translating}, this implies a construction of statistically-secure 1PRS in the SWAP model.
\end{proof}

\begin{theorem}
    If 1PRS exist, then IV-OWSG exists.
\end{theorem}

\begin{proof}
    This follows by combining several existing results in the literature.
    First, Theorem 1.1 of~\cite{morimae2022quantum} shows that if 1PRS exist, then non-interactive quantum commitments with computational hiding and statistical binding exist.
    Note that non-interactive quantum commitments trivially imply interactive quantum commitments.
    Next, Theorem 4.6 of~\cite{brakerski2022computational} shows that if quantum commitments with computational hiding and statistical binding exist, then EFI exists.
    Finally, Theorem 4.3 of~\cite{malavolta2024exponential} shows that if EFI exists, then IV-OWSGs exist.
\end{proof}

\bibliographystyle{alpha}
\bibliography{refs}

@article{hiroka2024computational,
  title={Computational complexity of learning efficiently generatable pure states},
  author={Hiroka, Taiga and Hsieh, Min-Hsiu},
  journal={arXiv preprint arXiv:2410.04373},
  year={2024}
}

@inproceedings{batra2024commitments,
  title={Commitments are equivalent to statistically-verifiable one-way state generators},
  author={Batra, Rishabh and Jain, Rahul},
  booktitle={2024 IEEE 65th Annual Symposium on Foundations of Computer Science (FOCS)},
  pages={1178--1192},
  year={2024},
  organization={IEEE}
}

@article{morimae2022one,
  title={One-wayness in quantum cryptography},
  author={Morimae, Tomoyuki and Yamakawa, Takashi},
  journal={arXiv preprint arXiv:2210.03394},
  year={2022}
}

@inproceedings{morimae2022quantum,
  title={Quantum commitments and signatures without one-way functions},
  author={Morimae, Tomoyuki and Yamakawa, Takashi},
  booktitle={Annual International Cryptology Conference},
  pages={269--295},
  year={2022},
  organization={Springer}
}

@article{brakerski2022computational,
  title={On the computational hardness needed for quantum cryptography},
  author={Brakerski, Zvika and Canetti, Ran and Qian, Luowen},
  journal={arXiv preprint arXiv:2209.04101},
  year={2022}
}

@inproceedings{malavolta2024exponential,
  title={Exponential Quantum One-Wayness and EFI Pairs},
  author={Malavolta, Giulio and Morimae, Tomoyuki and Walter, Michael and Yamakawa, Takashi},
  booktitle={International Conference on Security and Cryptography for Networks},
  pages={121--138},
  year={2024},
  organization={Springer}
}

@inproceedings{buadescu2019quantum,
  title={Quantum state certification},
  author={B{\u{a}}descu, Costin and O'Donnell, Ryan and Wright, John},
  booktitle={Proceedings of the 51st Annual ACM SIGACT Symposium on Theory of Computing},
  pages={503--514},
  year={2019}
}

@article{fefferman2025hardness,
  title={The Hardness of Learning Quantum Circuits and its Cryptographic Applications},
  author={Fefferman, Bill and Ghosh, Soumik and Sinha, Makrand and Yuen, Henry},
  journal={arXiv preprint arXiv:2504.15343},
  year={2025}
}

@inproceedings{rivest1991cryptography,
  title={Cryptography and machine learning},
  author={Rivest, Ronald L},
  booktitle={International Conference on the Theory and Application of Cryptology},
  pages={427--439},
  year={1991},
  organization={Springer}
}

@article{kearns1994cryptographic,
  title={Cryptographic limitations on learning boolean formulae and finite automata},
  author={Kearns, Michael and Valiant, Leslie},
  journal={Journal of the ACM (JACM)},
  volume={41},
  number={1},
  pages={67--95},
  year={1994},
  publisher={ACM New York, NY, USA}
}

@article{klivans2009cryptographic,
  title={Cryptographic hardness for learning intersections of halfspaces},
  author={Klivans, Adam R and Sherstov, Alexander A},
  journal={Journal of Computer and System Sciences},
  volume={75},
  number={1},
  pages={2--12},
  year={2009},
  publisher={Elsevier}
}

@article{song2021cryptographic,
  title={On the cryptographic hardness of learning single periodic neurons},
  author={Song, Min Jae and Zadik, Ilias and Bruna, Joan},
  journal={Advances in neural information processing systems},
  volume={34},
  pages={29602--29615},
  year={2021}
}

@inproceedings{daniely2014average,
  title={From average case complexity to improper learning complexity},
  author={Daniely, Amit and Linial, Nati and Shalev-Shwartz, Shai},
  booktitle={Proceedings of the forty-sixth annual ACM symposium on Theory of computing},
  pages={441--448},
  year={2014}
}

@inproceedings{hirahara2023learning,
  title={Learning in pessiland via inductive inference},
  author={Hirahara, Shuichi and Nanashima, Mikito},
  booktitle={2023 IEEE 64th Annual Symposium on Foundations of Computer Science (FOCS)},
  pages={447--457},
  year={2023},
  organization={IEEE}
}

@inproceedings{naor2006learning,
  title={Learning to impersonate},
  author={Naor, Moni and Rothblum, Guy N},
  booktitle={Proceedings of the 23rd international conference on Machine learning},
  pages={649--656},
  year={2006}
}

@article{valiant1984theory,
  title={A theory of the learnable},
  author={Valiant, Leslie G},
  journal={Communications of the ACM},
  volume={27},
  number={11},
  pages={1134--1142},
  year={1984},
  publisher={ACM New York, NY, USA}
}

@inproceedings{naor2015bloom,
  title={Bloom filters in adversarial environments},
  author={Naor, Moni and Yogev, Eylon},
  booktitle={Annual Cryptology Conference},
  pages={565--584},
  year={2015},
  organization={Springer}
}

@inproceedings{impagliazzo1990no,
  title={No better ways to generate hard NP instances than picking uniformly at random},
  author={Impagliazzo, Russell and LA, Levin},
  booktitle={Proceedings [1990] 31st Annual Symposium on Foundations of Computer Science},
  pages={812--821},
  year={1990},
  organization={IEEE}
}

@article{oliveira2016conspiracies,
  title={Conspiracies between learning algorithms, circuit lower bounds and pseudorandomness},
  author={Oliveira, Igor C and Santhanam, Rahul},
  journal={arXiv preprint arXiv:1611.01190},
  year={2016}
}

@inproceedings{blum1993cryptographic,
  title={Cryptographic primitives based on hard learning problems},
  author={Blum, Avrim and Furst, Merrick and Kearns, Michael and Lipton, Richard J},
  booktitle={Annual international cryptology conference},
  pages={278--291},
  year={1993},
  organization={Springer}
}

@article{regev2009lattices,
  title={On lattices, learning with errors, random linear codes, and cryptography},
  author={Regev, Oded},
  journal={Journal of the ACM (JACM)},
  volume={56},
  number={6},
  pages={1--40},
  year={2009},
  publisher={ACM New York, NY, USA}
}

@incollection{micciancio2009lattice,
  title={Lattice-based cryptography},
  author={Micciancio, Daniele and Regev, Oded},
  booktitle={Post-quantum cryptography},
  pages={147--191},
  year={2009},
  publisher={Springer}
}

@inproceedings{gentry2008trapdoors,
  title={Trapdoors for hard lattices and new cryptographic constructions},
  author={Gentry, Craig and Peikert, Chris and Vaikuntanathan, Vinod},
  booktitle={Proceedings of the fortieth annual ACM symposium on Theory of computing},
  pages={197--206},
  year={2008}
}

@article{angluin1988queries,
  title={Queries and concept learning},
  author={Angluin, Dana},
  journal={Machine learning},
  volume={2},
  pages={319--342},
  year={1988},
  publisher={Springer}
}

@inproceedings{lyubashevsky2010ideal,
  title={On ideal lattices and learning with errors over rings},
  author={Lyubashevsky, Vadim and Peikert, Chris and Regev, Oded},
  booktitle={Advances in Cryptology--EUROCRYPT 2010: 29th Annual International Conference on the Theory and Applications of Cryptographic Techniques, French Riviera, May 30--June 3, 2010. Proceedings 29},
  pages={1--23},
  year={2010},
  organization={Springer}
}

@article{zhao2024learning,
  title={Learning quantum states and unitaries of bounded gate complexity},
  author={Zhao, Haimeng and Lewis, Laura and Kannan, Ishaan and Quek, Yihui and Huang, Hsin-Yuan and Caro, Matthias C},
  journal={PRX Quantum},
  volume={5},
  number={4},
  pages={040306},
  year={2024},
  publisher={APS}
}

@article{yang2023complexity,
  title={The Complexity of Learning (Pseudo) random Dynamics of Black Holes and Other Chaotic Systems},
  author={Yang, Lisa and Engelhardt, Netta},
  journal={arXiv preprint arXiv:2302.11013},
  year={2023}
}

@article{foxman2025random,
  title={Random unitaries in constant (quantum) time},
  author={Foxman, Ben and Parham, Natalie and Vasconcelos, Francisca and Yuen, Henry},
  journal={arXiv preprint arXiv:2508.11487},
  year={2025}
}

@article{niroula2026digital,
  title={Digital signatures with classical shadows on near-term quantum computers},
  author={Niroula, Pradeep and Liu, Minzhao and Omanakuttan, Sivaprasad and Amaro, David and Chakrabarti, Shouvanik and Ghosh, Soumik and He, Zichang and Jin, Yuwei and Kaleoglu, Fatih and Kordonowy, Steven and others},
  journal={arXiv preprint arXiv:2602.04859},
  year={2026}
}

@article{lewis2025computational,
  title={Computational complexity in quantum learning tasks},
  author={Lewis, Laura},
  journal={Edinburgh Research Archive, Informatics Thesis and Dissertation Collection},
  howpublished={\url{https://era.ed.ac.uk/items/40a54045-05d1-438a-9f05-c981bee5081a}},
  year={2025},
  publisher={The University of Edinburgh}
}

@inproceedings{buadescu2021improved,
  title={Improved quantum data analysis},
  author={B{\u{a}}descu, Costin and O'Donnell, Ryan},
  booktitle={Proceedings of the 53rd Annual ACM SIGACT Symposium on Theory of Computing},
  pages={1398--1411},
  year={2021}
}

@inproceedings{goldin2025translating,
  title={Translating between the common haar random state model and the unitary model},
  author={Goldin, Eli and Zhandry, Mark},
  booktitle={Annual International Cryptology Conference},
  pages={269--300},
  year={2025},
  organization={Springer}
}

@inproceedings{chen2025power,
  title={The power of a single haar random state: constructing and separating quantum pseudorandomness},
  author={Chen, Boyang and Coladangelo, Andrea and Sattath, Or},
  booktitle={Annual International Conference on the Theory and Applications of Cryptographic Techniques},
  pages={108--137},
  year={2025},
  organization={Springer}
}

@article{hiroka2025hardness,
  title={Hardness of Quantum Distribution Learning and Quantum Cryptography},
  author={Hiroka, Taiga and Hsieh, Min-Hsiu and Morimae, Tomoyuki},
  journal={arXiv preprint arXiv:2507.01292},
  year={2025}
}

@misc{morimae25collapsing,
      title={Quantum Cryptography and Hardness of Non-Collapsing Measurements}, 
      author={Tomoyuki Morimae and Yuki Shirakawa and Takashi Yamakawa},
      year={2025},
      eprint={2510.04448},
      archivePrefix={arXiv},
      primaryClass={quant-ph},
      url={https://arxiv.org/abs/2510.04448}, 
}

\end{document}